\documentclass[times,preprint,10pt]{elsarticle}
\usepackage{amssymb}
\usepackage{amsmath}

\journal{Annals of Physics}

\usepackage{graphicx}  % needed for figures
\usepackage{dcolumn}   % needed for some tables
\usepackage{bm}        % for math
\usepackage{amssymb}   % for math
\usepackage{amsbsy}
\usepackage{amsmath}
\usepackage{amsfonts}
\usepackage{color}
\usepackage[dvipsnames]{xcolor}
\usepackage{amsthm}
\usepackage{graphics}
\usepackage{graphicx}
\usepackage{xcolor}
\usepackage{caption}
\usepackage{subcaption}
\usepackage{amssymb}
\usepackage{enumitem}%indent in itemize
\setlist[itemize]{leftmargin=*}%Noindent in itemize
\biboptions{numbers,sort&compress}

\theoremstyle{definition}
\newtheorem{prop}{Proposition}%[teo] se usa para que siga el mismo contador que los teoremas.
\theoremstyle{definition}

\theoremstyle{definition}	 %AÃ±adimos esto para que tenga el estilo de una definiciÃ³n y no la de teorema o proposiciÃ³n. Por defecto el estilo es "plain", que es el de teoremas.
\theoremstyle{definition}
\theoremstyle{definition} %AÃ±adimos esto para que tenga el estilo de una definiciÃ³n y no la de teorema o proposiciÃ³n. Por defecto el estilo es "plain", que es el de teoremas.

\theoremstyle{definition} %El estilo remark podrÃ­a ponerse aquÃ­ tambiÃ©n

\theoremstyle{definition} %El estilo remark podrÃ­a ponerse aquÃ­ tambiÃ©n

\renewenvironment{proof}{{\bfseries \noindent Proof.\hspace{3mm}}}{\qed\\} %Cambia el proof para que estÃ© en negrita

\makeatletter
\renewcommand{\p@subsection}{}
\renewcommand{\p@subsubsection}{}
\makeatother
\usepackage{xcolor}
\usepackage{graphicx}
\definecolor{title}{RGB}{0,0,0}
\definecolor{other}{RGB}{0,0,0}
\definecolor{name}{RGB}{0,0,0}
\definecolor{phd}{RGB}{0,0,0}
\usepackage[hidelinks]{hyperref}
\hypersetup{
  colorlinks   = true, 	%Colours links instead of ugly boxes
  urlcolor     = blue, 	%Colour for external hyperlinks
  linkcolor    = blue, 	%Colour of internal links
  citecolor   = red 	%Colour of citations
}

\definecolor{applegreen}{rgb}{0.0, 0.71, 0.0}

\newcommand{\hvu}{\boldsymbol{\mathcal{H}}_{0}}

\newcommand{\hvup}{\boldsymbol{\mathcal{H}}_{_{V_0}}}
\begin{document}
\begin{frontmatter}

\title{Hyperspherical ${\delta\text{-}\delta^\prime}$ potentials}
\author[valladolid]{J. M. Mu{\~n}oz-Casta{\~n}eda}\ead{jose.munoz.castaneda@uva.es}

\author[valladolid,imuva]{L.M. Nieto}
\ead{luismiguel.nieto.calzada@uva.es}
\author[valladolid]{C. Romaniega}
\ead{cesar.romaniega@alumnos.uva.es}
%\cortext[cor2]{Principal corresponding author}

%
\address[valladolid]{Departamento de F\'{\i}sica Te\'{o}rica, At\'{o}mica y \'{O}ptica,
Universidad de Valladolid, 47011 Valladolid, Spain}
\address[imuva]{IMUVA- Instituto de Matematicas, Universidad de Valladolid, 47011 Valladolid, Spain}
%
%\address[potsdam]{Department of Physics, Clarkson University, Potsdam, NY 13699, USA}
%\address[cerfim]{Departamento de F\'{\i}sica Te\'{o}rica II, UCM, Madrid, Spain}

\begin{abstract}
The spherically symmetric potential $a \,\delta  (r-r_0)+b\,\delta ' (r-r_0)$ is generalised for the $d$-dimensional space as a characterisation of a unique selfadjoint extension of the free Hamiltonian. For this extension of the Dirac delta, the spectrum of negative, zero and positive energy states is studied in $d\geq 2$, providing numerical results for the expectation value of the radius as a function of the free parameters of the potential. Remarkably, only if $d=2$ the $\delta$-$\delta'$ potential for arbitrary $a>0$ admits a bound state with zero angular momentum.

\end{abstract}


\end{frontmatter}
%%%%%%%%%%%%%%%%%%%%%%%%%%%%%%%%%%%%%%%%%%%
\section{Introduction}
%%%%%%%%%%%%%%%%%%%%%%%%%%%%%%%%%%%%%%%%%%%
The presence of boundaries has played a central role in many areas of physics for many years. In this respect, concerning the quantum world, one of the most significant phenomenon is due to the interaction of quantum vacuum fluctuations of the electromagnetic field with two conducting ideal plane parallel plates: the Casimir effect \cite{casimir}, meassured by Sparnaay in 1959~\cite{sparnaay58}. Frontiers are also essential in the theory of quantum black holes, where one of the most remarkable results is the brick wall model developed by G. 't Hooft \cite{thooft1,thooft2}, in which boundary conditions are used to implement the interaction of quantum massless particles with the black hole horizon observed from far away. In addition, the propagation of plasmons over the graphene sheet and the surprising scattering properties through abrupt defects \cite{lmm} can be understood by using boundary conditions to represent the defects. In all these situations, the physical properties of the frontiers and their interaction with quantum objects of the bulk are mimicked by different boundary conditions. Many of these effects concerning condensed matter quantum field theory can be reproduced in the laboratory.

Moreover, point potentials or potentials supported on a point have attracted much attention over the years (see \cite{belloni14rev} for a review). These kind of potentials, also called contact interactions, enables us to build integrable toy-model approximations for very localised interactions. The most known example of such kind of interaction is the  Dirac-$\delta$ potential, that has been extensively studied in the literature (see, e.g. \cite{kronig-prsa31,galindo1990quantum,jackiw-deltas}). Since this potential admits only one bound state when it has negative coupling \cite{galindo1990quantum,jackiw-deltas}, it can represent Hydrogen-like nuclei in interaction with a classical background. Dirac-$\delta$ potentials can also be used to represent extended plates in effective scalar quantum field theories to compute the quantum vacuum interaction for semi-transparent plates in flat spacetime or  curved backgrounds \cite{munoz2015delta,JMC,JMG,bordag-jpa25}.

From what has been mentioned above, it is intiutively clear that quantum boundaries and contact interactions are almost the same. The rigorous mathematical framework to study them is the use of selfadjoint extensions to represent extended objects and point supported potentials (see \cite{AIM2,AIM1,AGLM} for a more physical point of view), such as plates in the Casimir effect setup \cite{aso-mc1,aso-mc2}, or contact interactions more general than the Dirac-$\delta$ in quantum mechanics and effective quantum field theories \cite{kurasov1996distribution,dittrich92,gitman2012self,albeverio-book,ibort1,fucci1,fucci2}. The theory of selfadjoint extensions for symmetric operators has been well known to mathematicians for many years. However, it only became a valuable tool for modern quantum physics after the seminal works of Asorey {\it et al} \cite{AIM1,aso-mc1,aso-mc2}, in which the problem was re-formulated in terms of physically meaningful quantities for relevant operators in quantum mechanics and quantum field theory \cite{trg1,trg2,fucci3,fucci4,fassari18,facchi18}.

One of the most immediate extensions of the Dirac-$\delta$ potential is the $\delta^\prime$-potential
$
V_{\delta^\prime}=b\delta^\prime (x).
$
Over the last years, there has been some controversy about the definition of this potential 
in one dimension (see the discussion in \cite{gadella2009bound,gadella-jpa16}), and yet it is not clear how $V_{\delta^\prime}$ should be characterised. The aim of this paper is not to discuss this definition but to use the one introduced in \cite{gadella2009bound}, including a Dirac-$\delta$ to regularise the potential, and study its generalisation as a hyperspherical potential in dimension $d>1$. We will fully solve the non-relativistic quantum mechanical problem associated with the spherically symmetric potential
\begin{equation}\label{vddp}
\widehat V_{\delta\text{-}\delta^\prime}(r) =  a \,\delta  (r-r_0)+b\,\delta ' (r-r_0),\quad a,b\in \mathbb{R},\,\,r_0>0.
\end{equation}
Due to the radial symmetry of the problem, we will end up having a family of one-dimensional Hamiltonians (the radial Hamiltonian), for which a generalisation of the definition given in \cite{gadella2009bound,kurasov1996distribution} is needed.

This paper is organised as follows. Section \ref{sec:2} defines the spherically symmetric $\delta$-$\delta'$  potential in arbitrary dimension based on the work for one dimensional systems performed in \cite{kurasov1996distribution}. Having determined the  properties which characterise the potential, we carry out a thorough study of the bound states structure in Section \ref{sec:BS} and of the zero-modes and scattering states in Section~\ref{zmode-sec}. In the latter, we also compute some numerical results concerning the mean value of the position (radius) operator. Through these two sections we specially focus on the peculiarities of the two dimensional case. Finally, in Section \ref{sec:conclu}  we present our concluding remarks. 

%%%%%%%%%%%%%%%%%%%%%%%%%%%%%%%%%%%%%%%%%%%
\section{The $\delta\text{-}\delta^\prime$ interaction in the $d$-dimensional Schr\"odinger equation}\label{sec:2}
%%%%%%%%%%%%%%%%%%%%%%%%%%%%%%%%%%%%%%%%%%%

We consider a non-relativistic  quantum particle of mass $m$ moving  in $\mathbb{R}^d $ ($d=2,3,\dots$) under the influence of the spherically symmetric potential $  \widehat V_{\delta\text{-}\delta^\prime}(r)$ given in \eqref{vddp}. 
The quantum Hamiltonian operator that governs the dynamics of the system is
\begin{equation}
{{\bf H}}=\frac{-\hbar^2}{2\,m}\widehat\Delta_d+ \widehat V_{\delta\text{-}\delta^\prime}(r),
\end{equation}
where $\widehat\Delta_d$ is the $d$-dimensional Laplace operator.
To start with, let us analyse the dimensions of the free parameters $a$ and $b$ that appear in our system. Using the properties of the Dirac-$\delta$ under dilatations and knowing that the $\delta^\prime$ has to have the same units as the formal expression ${d\delta(x)}/{dx}$  it is straightforward to see that the dimensions of the parameters $a$ and $b$ are
\begin{equation}
[a]=L^3T^{-2}M,\quad [b]=L^4T^{-2}M.
\end{equation}
Hence, we can introduce the following dimensionless quantities:
\begin{equation}
{{\bf h}}\equiv\frac{2}{mc^2}{{\bf H}},\quad w_0\equiv\frac{2a}{\hbar c},\quad w_1\equiv\frac{b m}{\hbar^2},\quad 
 x\equiv\frac{mc}{\hbar}r .
\end{equation}
With the previous definitions, the dimensionless quantum Hamiltonian reads
\begin{equation}
{{\bf h}}=-\Delta_d+ w_0 \,\delta  (x-x_0)+2w_1\,\delta ' (x-x_0).\label{5}
\end{equation}
Introducing hyperspherical coordinates, $(x,\Omega_d\equiv \lbrace\theta_1,\dots,\theta_{d-2},\phi\rbrace)$, the $d$-dimensional Laplace operator $\Delta_d$ is written as
\begin{equation}\label{eq:LaplacianOperator}
\Delta_d=\frac{1}{x^{d-1}}\frac{\partial}{\partial x}\left(x^{d-1}\frac{\partial}{\partial x}\right)+\frac{\Delta_{S^{d-1}}}{x^2},
\end{equation}
where $\Delta_{S^{d-1}}=-{\bf L}_{d}^{2}$ is the Laplace-Beltrami operator in the hypersphere $S^{d-1}$, and minus the square of the generalised dimensionless angular momentum operator \cite{avery2010harmonic}. In hyperspherical coordinates, the eigenvalue equation for ${{\bf h}}$ in \eqref{5} is separable, and therefore we can write the solutions as
\begin{equation}\label{7}
\psi_{\lambda\ell}(x,\Omega_d)=R_{\lambda\ell}(x) Y_\ell(\Omega_d),
\end{equation}
where $R_{\lambda\ell}(x)$ is the radial wave function and $Y_\ell(\Omega_d)$ are the hyperspherical harmonics which are the eigenfunctions of $\Delta_{S^{d-1}}$ with eigenvalue (see \cite{kirsten-book} and references therein)
\begin{equation}
\chi(d,\ell)\equiv-\ell(\ell+d-2).\label{8}
\end{equation}
The degeneracy of  $\chi(\ell,d)$ is given by \cite{seto1974bargmann}
\begin{equation}\label{9}
\text{deg}(d,\ell) =\begin{cases}
\dfrac{  (d+\ell-3)!}{(d-2)! \ell !}(d+2 (\ell-1))&\text{if}\quad d\neq 2\ \ \text{and}\ \ \ell\neq 0,\\[1.5ex]
\hfil 1 &\text{if}\quad d= 2\ \ \text{and}\ \ \ell=0.
\end{cases}
\end{equation}
In three dimensions we come up with $\chi(3,\ell)=-\ell(\ell+1)$ and $\text{deg}(3,\ell) =2\ell +1$ as expected.
Taking into account the eigenvalue equation for \eqref{5} and equations (\ref{7},\,\ref{8}) the radial wave function fulfils
 \begin{equation}\label{eq11}   
\left[- \dfrac{d^2}{dx^2}\! -\!\dfrac{d-1}{x}\dfrac{d}{dx}
 + \dfrac{\ell(\ell+d-2)}{x^2}+V_{\delta\text{-}\delta^\prime}(x) \right]{R}_{\lambda\ell}(x)=\lambda R_{\lambda\ell},
\end{equation}  
being 
\begin{equation}\label{eq:PotentialDeltaDeltap}
V_{\delta\text{-}\delta^\prime}(x)=w_0\delta(x-x_0)+2w_1\delta^\prime(x-x_0).
\end{equation}

To solve the eigenvalue equation \eqref{eq11}, we first need to define the potential $V_{\delta\text{-}\delta^\prime}$. In order to characterise the potential $V_{\delta\text{-}\delta^\prime}(x)$ as a selfadjoint extension following \cite{gadella2009bound,kurasov1996distribution}, 
we introduce the reduced radial function
\begin{equation}\label{eq:reducedfunction}
u_{\lambda\ell}(x) \equiv x^{\frac{d-1}{2}}{R}_{\lambda\ell}(x),
\end{equation}
to remove the first derivative from the one dimensional radial operators in  \eqref{eq11}. Taking into account \eqref{eq11} and \eqref{eq:reducedfunction}, we obtain the eigenvalue problem that this function satisfies
\begin{equation}\label{eq:ReducedRadialEq}
\left(\hvu+V_{\delta\text{-}\delta^\prime}(x)\right) u_{\lambda\ell}(x)=\lambda_\ell u_{\lambda\ell}(x),
\end{equation}
where
\begin{equation}\label{16}
\hvu\equiv -\dfrac{d^2}{dx^2}\! +
 \frac{(d+2 \ell-3) (d+2 \ell-1)}{4 x^2} .
\end{equation}
Thus, as in \cite{kurasov1996distribution}, we define the potential $V_{\delta\text{-}\delta^\prime}$ through a set of matching conditions  on the eigenfuntion of $\hvu$ at $x=x_0^\pm$. The discussion for the one dimensional case imposes that the wave function $\psi$ must belong to the Sobolev space $W^2_2(\mathbb{R}\setminus\{0\})$ in order to ensure that $\psi''(x) $ is a square integrable function, i.e., the mean value of the kinetic energy is finite. To generalise this condition to higher dimensional Hamiltonians with spherical symmetry, we need to impose that the domain of wave functions where the operator $\hvu$ is selfadjoint when it is defined on $\mathbb{R}_{>0}$ is 
\begin{equation}
W\left(\hvu,\mathbb{R}_{>0}\right)\equiv\{ f(x)\in L^2(\mathbb{R}_{>0})\vert\,\langle \hvu \rangle_{f(x)}<\infty\},
\end{equation}
where the expectation value of $\hvu$ is defined as usual  $$\langle \hvu \rangle_{f(x)}\equiv\int_0^\infty  f^*(x)\left( \hvu f(x)\right)\,dx.$$
When we remove  the point $x=x_0$ the operator $\hvu$ is no longer selfadjoint on the space of functions $W\left(\hvu,\mathbb{R}_{x_0}\right)\equiv\{ f(x)\in L^2(\mathbb{R}_{x_0})\vert\,\langle \hvu \rangle_{f(x)}<\infty\}$ since
\begin{eqnarray*}
\int_0^\infty\!\!\!\! dx\, \phi^*\left(\hvu\varphi\right)-\int_0^\infty\!\!\!\!  dx\, \varphi\left(\hvu\phi\right)^*\neq 0,\quad \varphi,\phi\in W\left(\hvu,\mathbb{R}_{x_0}\right),
\end{eqnarray*}
due to the boundary terms appearing when integrating by parts twice. Nevertheless, $\hvu$ is symmetric on the subspace  given by the closure of the $L^2(\mathbb{R}_{x_0})$ functions with compact support in $\mathbb{R}_{x_0}$. This situation generalises the initial conditions given in \cite{kurasov1996distribution}, and matches the geometric view in \cite{AIM1,AIM2}.
% Under these conditions, the argument that determines the matching conditions described in \cite{kurasov1996distribution} can be used in exactly the same way for the Hamiltonian $\hvu$ defined over $W\left(\hvu,\mathbb{R}_{x_0}\right)$. 
Hence, the domain of the selfadjoint extension $\hvu+V_{\delta\text{-}\delta^\prime}$ of the operator $\hvu$ defined on  $\mathbb{R}_{x_0}$ is given by
\begin{equation}\label{eq:MatchingU}
 {\cal D}\left( \hvu+V_{\delta\text{-}\delta^\prime}\right)=\left\{ f\in W\left(\hvu,\mathbb{R}_{x_0}\right)\vert\,\left(
\begin{array}{c}
 f(x_0^{+}) \\
 {f'(x_0^{+})} \\
\end{array}
\right)=\left(
\begin{array}{cc}
 \alpha  & 0 \\
 \beta  & {\alpha^{-1} } \\
\end{array}
\right)\left(
\begin{array}{c}
 f(x_0^{-}) \\
 f'(x_0^{-}) \\
\end{array}
\right)\right\}, 
\end{equation}
where we have introduced the values
\begin{equation}\label{eq:Notation}
\alpha\equiv\frac{1+w_1}{1-w_1},\quad \beta\equiv \frac{w_0}{1-w_1^2}.
\end{equation}
%Throughout the paper, we will use a prime to denote the derivative with respect to the argument of the function. 
Now, using \eqref{eq:reducedfunction} in \eqref{eq:MatchingU} we obtain the following matching conditions for the radial wave function $R_{\lambda\ell}$:
\begin{equation}\label{eq:MatchingR}
\left(
\begin{array}{c}
 R_{\lambda\ell} (x_0^{+}) \\
 R'_{\lambda\ell}(x_0^{+}) \\
\end{array}
\right)= \left(
\begin{array}{cc}
 \alpha  & 0 \\
 {\tilde \beta}_{} & {\alpha^{-1} } \\
\end{array}
\right)\left(
\begin{array}{c}
R_{\lambda\ell} (x_0^{-}) \\
  R'_{\lambda\ell}(x_0^{-}) \\
\end{array}
\right), 
\end{equation}
where the effective couplings ${\tilde \beta}$ and ${\tilde w_0}$ are
\begin{eqnarray}\label{eq:NewCouplings}
{\tilde \beta} \equiv  \beta -\frac{\left(\alpha ^2-1\right) (d-1)}{2 \, \alpha\,  x_0} =  \frac{{\tilde w_0}}{1-w_{1}^2}\quad \Rightarrow \quad
{\tilde w_0}\equiv \dfrac{2 (1-d)\, {w_{1}}}{x_0}+{w_0}.
\end{eqnarray}
Observe that when we turn off the $\delta'$ contribution, $w_1=0$ or $\alpha=1$,  the  finite discontinuity in the derivative that characterises the $\delta$-potential arises
\begin{equation*}
R_{\lambda\ell} (x_0^{+})=R_{\lambda\ell} (x_0^{-}) \quad \text{and}\quad
 R'_{\lambda\ell}(x_0^{+}) - R'_{\lambda\ell}(x_0^{-}) = w_0 R_{\lambda\ell} (x_0).
\end{equation*}
On the other hand, when $w_1 = \pm 1$ the matching condition matrix  is ill defined because it does not relate the boundary data on $x_0^-$ with those on $x_0^+$. This case is treated in detail in \cite{munoz2015delta}, where it is demonstrated that $w_1=\pm1$ leads to Robin and Dirichlet boundary conditions in each side of the singularity $x=x_0$. Specifically, 
\begin{equation}\label{eq:RobinDirichlet}
\begin{aligned}
{R_{\lambda\ell}}(x_0^+)-\frac{4}{{\tilde w_0}^+}R_{\lambda\ell}^\prime(x_0^+)&=0, \qquad {R_{\lambda\ell}}(x_0^-)=0\quad\text{if}\quad w_1= 1,\\[0.5ex]
{R_{\lambda\ell}}(x_0^-)+\frac{4}{{\tilde w_0}^-}R_{\lambda\ell}^\prime(x_0^-)&=0, \qquad {R_{\lambda\ell}}(x_0^+)=0\quad\text{if}\quad w_1= -1,
\end{aligned}
\end{equation}
where ${\tilde w_0}^{\pm}= w_0\pm{2 (1-d)}/{x_0}$. Recently the potential \eqref{eq:PotentialDeltaDeltap} was studied for two and three dimensions in \cite{govindarajan2016modelling} where the matching conditions used for $R_{\lambda\ell}$ are those in \eqref{eq:MatchingU} instead of \eqref{eq:MatchingR}, an approximation which is only satisfied if
\begin{equation}\label{eq:X}
 x_0\,\vert w_0\vert \gg\vert w_1\vert
\end{equation}
Throughout the text, we will point out the equations that are valid even when the previous inequality does not hold.

\subsubsection*{A remark on selfadjoint extensions and point supported potentials}  The operator $\hvu$ defined as a one dimensional Hamiltonian over the physical space $\mathbb{R}_{x_0}$ is not selfadjoint as we have seen. In order to define a true Hamiltonian as a selfadjoint operator one has to select a selfadjoint extension of $\hvu$ in the way explained above for the particular case of the potential $V_{\delta\text{-}\delta^\prime}(x)$. More generally, the set of all selfadjoint extensions is in one-to-one correspondence with the set of unitary matrices $U(2)$. As was demonstrated in \cite{AIM1}, for a given unitary matrix $G\in U(2)$ there is a unique selfajoint extension $\hvu^{G}$ of $\hvu$. In this sense, the selfadjoint extension $\hvu^{G}$ can be thought in a more physically meaningful way as a potential $V_G(x-x_0)$ supported on a point $x_0$ for the quantum Hamiltonian $\hvu$ and write $\hvu+V_G(x-x_0)\equiv\hvu^{G}$. Physically one would just think on $V_G(x-x_0)$ as a potential term in the same way as the Dirac-$\delta$ potential \cite{galindo1990quantum}. In this view, once the operator $\hvu$ is fixed, the selfadjoint extensions can be seen as potentials supported on a point, and the other way around because of the one-to-one correspondence demonstrated in \cite{AIM1} (and recently reviewed in \cite{AIM2}).

%%%%%%%%%%%%%%%%%%%%%%%%%%%%%%%%%%%%%%%%%%%
\section{Bound states with the free Hamiltonian and the singular interaction}\label{sec:BS}
%%%%%%%%%%%%%%%%%%%%%%%%%%%%%%%%%%%%%%%%%%%

In this section we will analyse in detail the discrete spectrum of negative energy states (bound states) for the $\delta\text{-}\delta^\prime$ potential. In particular, we will give an analytic formula for the number of them as a function of the parameters $\{w_0,w_1,x_0\}$. As the eigenvalue equation for the bound states is \eqref{eq11} with $\lambda <0$, we define $\lambda\equiv-\kappa^2$ with $\kappa>0$, and replace the subindex $\lambda$ by $\kappa$ in the wave functions all over this section.
The general form of the solutions of  equation \eqref{eq11}  is
\begin{eqnarray}\label{eq:RadialB}
{R}_{\kappa\ell}(x)   = \begin{cases}A_1\, \mathcal{I}_{\ell}(\kappa x) +B_1\,\mathcal{K}_{\ell}(\kappa x) \quad \text{if}\quad x\in (0, x_0),\\[0.5ex] 
A_2\,\mathcal{I}_{\ell}(\kappa x) +B_2\,\mathcal{K}_{\ell}(\kappa x) \quad \text{if}\quad x\in (x_0, \infty), 
\end{cases}
\end{eqnarray}
being $\mathcal{I}_{\ell}(z)$ and $\mathcal{K}_{\ell}(z)$, up to a constant factor, the  modified hyperspherical Bessel functions of the first and second kind respectively
\begin{equation}
\mathcal{I}_{\ell}(z)\equiv \dfrac{1}{z^{\nu}} \, I_{\ell+\nu}( x), \quad\mathcal{K}_{\ell}(z)\equiv \dfrac{1}{z^{\nu}} \, K_{\ell+\nu}(z)\quad\text{with}\quad \nu\equiv\frac{d-2}{2},
\end{equation}
Similarly from Eq.\eqref{eq:reducedfunction} the general form of the reduced radial function is
\begin{equation}\label{eq:ReducedRadialB}
{u}_{\kappa\ell}(x)   =\sqrt{x} \begin{cases}A_1\, I_{\ell+\nu}(\kappa x) +B_1\,K_{\ell+\nu}(\kappa x)  \quad  \text{if}\quad x\in (0, x_0),\\[0.5ex] 
A_2\, I_{\ell+\nu}(\kappa x) +B_2\,K_{\ell+\nu}(\kappa x)  \quad  \text{if}\quad x\in (x_0, \infty).
\end{cases}
\end{equation}
The integrability condition on the reduced radial function
\begin{equation*}
\int_0^\infty |{u}_{\kappa\ell}(x)|^2\,dx <\infty
\end{equation*}
imposes $A_2 = 0$. Moreover, the solution multiplied by $B_1$ is not square integrable  except for zero angular momentum in two and three dmensions  \cite{olver2010nist}. The regularity condition of the wave function at the origin
$
u_{\kappa \ell}(x=0)=0,
$ 
sets $B_1=0$ for $d=3$. It would seem that the two solutions in the inner region are admissible when $d=2$, but  $B_1\neq 0$ would lead to a normalizable  bound state with arbitrary negative
energy \cite{berry1973semiclassical}. In addition, for any  wave function $ \psi $, the following identity involving the mean value of the kinetic energy operator:
\begin{equation}\label{eq:KineticOperator}
\dfrac{1}{2\,m}\langle\psi\vert P^2\vert\psi\rangle = \dfrac{1}{2\,m}(\langle\psi\vert \boldsymbol{P})\cdot(\boldsymbol{P}\vert\psi\rangle),
\end{equation}
holds if we impose certain conditions on the wave function at the boundary $x=0$, which are not satisfied by $\mathcal{K}_0$. Hence, we conclude that $B_1$ should be zero for all the cases.
With the previous analysis and (\ref{eq:MatchingR})   we obtain the matching condition 
\begin{equation}\label{eq:RecurrenceRelation}
B_2 \left(
\begin{array}{c}
 \mathcal{K}_{\ell}(\kappa x_0) \\[0.5ex]
 \left.\kappa\, \mathcal{K}'_{\ell}(\kappa x_0)\right.\\
\end{array}
\!\right) = A_1 \left(
\begin{array}{cc}
 \alpha  & 0 \\[0.5ex]
 {\tilde \beta}_{} & {\alpha^{-1} } \\
\end{array}
\right)  \left(
\begin{array}{c}
 \mathcal{I}_{\ell}(\kappa x_0) \\[0.5ex] 
 \left.\kappa\, \mathcal{I}'_{\ell}(\kappa x_0)\right.\\
\end{array}
\!\right)\!, 
\end{equation} 
from which  the secular equation is obtained
\begin{equation}\label{eq:SECULAR}
\alpha\left.\dfrac{d}{dx}\log \mathcal{K}_{\ell}(\kappa x) \right\vert_{x=x_0}= {{\tilde \beta} + \alpha^{-1}\, \left.\dfrac{d}{dx}\log \mathcal{I}_{\ell}(\kappa x)\right\vert_{x=x_0}}{ }.
\end{equation}
The solutions for $\kappa>0$ of the previous equation give the energies of the bound states accounting for $\lambda=-\kappa^2$. The equation \eqref{eq:SECULAR} can be written as
\begin{equation}\label{eq30}
F(y_0)\equiv -y_0 \left(\frac{I_{\nu+\ell-1}(y_0)}{\alpha \, I_{\nu+\ell}(y_0)}+\frac{\alpha\,  K_{\nu+\ell-1}(y_0)}{K_{\nu+\ell}(y_0)}\right)-\left(\alpha -\alpha^{-1}\right)\ell= {2\nu\left(\alpha -\alpha^{-1}\right) }+{\tilde \beta}  x_0,
\end{equation}
where $y_0\equiv\kappa x_0$ and the right hand side is independent of the  energy and the angular momentum. For  $d=2,3$ the  results of \cite{govindarajan2016modelling} are obtained as a limiting case ($x_0\vert w_0\vert\gg \vert w_1\vert$). In particular, the secular equation for the  $\delta$-potential ($\alpha = 1$ and ${\tilde \beta}=w_0$) is
\begin{equation*}
-y_0 \left(\frac{I_{\nu+\ell-1}(y_0)}{  I_{\nu+\ell}(y_0)}+\frac{K_{\nu+\ell-1}(y_0)}{K_{\nu+\ell}(y_0)}\right)= w_0\,  x_0.
\end{equation*}

\subsection{On the number of bound states}\label{sec:NumberBS}
Although equation \eqref{eq30} can not be solved analytically in $\kappa$, it can be used to characterise some fundamental aspects of the set of positive solutions of \eqref{eq30}. The main feature is the number of bound states that exist for $d$ and $\ell$
\begin{equation*}
N^d_\ell = n^d_\ell\, \text{deg}(d, \ell),
\end{equation*}
where $n^d_\ell$ is the the number of negative energy eigenvalues and  $\text{deg}(d, \ell)$ is the degeneracy associated with $\ell$ in $d$ dimensions \eqref{9}. In this way, we first delimit the possible values of $n^d_\ell$.
\begin{prop}\label{res:EnergyAngularMomenta}
In the $d$-dimensional quantum system described by the Hamiltonian \eqref{5}
the number $n^d_\ell$  is at most one, i.e., 
$
n^d_\ell \in \{ 0,1\}.
$

\begin{proof}{\text{}}
From \eqref{eq30} and applying the properties of the Bessel functions, the derivative of $F(\kappa x_0)$ with respect to $\kappa$ is
\begin{equation*}
x_0F'(y_0) = {y_0} \left[\alpha  \left(\frac{K_{\nu+\ell-1}(y_0) K_{\nu+\ell+1}(y_0)}{K_{\nu+\ell}(y_0)^2}-1\right)+\alpha^{-1}\left(1-\frac{I_{\nu+\ell-1}(y_0) I_{\nu+\ell+1}(y_0)}{I_{\nu+\ell}(y_0)^2}\right)\right],
\end{equation*}
and, as it is proven in \cite{laforgia2010some} and the references cited therein,
\begin{eqnarray*}
{K_{n-1}(y_0) K_{n+1}(y_0)}&>&{K_n(y_0)^2},\quad \text{if}\quad  y_0>0,\ n \geq -1/2,\\[0.5ex]
{I_{n-1}(y_0) I_{n+1}(y_0)}&<&{I_n(y_0)^2},\quad\ \,\text{if}\quad y_0>0,\ n\in \mathbb{R}.
\end{eqnarray*}
In the present case $n=\nu+\ell \geq 0$, therefore we can conclude that
\begin{equation}\label{eq:SignL1}
\text{sgn} \left(F'(y_0)\right) = \text{sgn} \left( \alpha\right).
\end{equation}
Hence, except for $\alpha = 0$ (ill defined matching conditions) $F(y_0)$ is a strictly monotone function and the proposition is proved.
\end{proof}
\end{prop}
\noindent This result is in agreement with the Bargmann's inequalities for a general potential in three dimensional systems
\begin{equation*}
n^{d=3}_\ell< \dfrac{1}{2 \,\ell +1}\int_0^\infty x |V(x)|\,dx,
\end{equation*}
which guarantees a finite number of bound states when the integral is convergent \cite{bargmann1952number}. Moreover, this inequality was generalised for arbitrary dimensional systems with spherical symmetry \cite{seto1974bargmann}
\begin{eqnarray}
n^{d}_\ell< \dfrac{1}{2 \,\ell +d-2}\int_0^\infty x |V(x)|\, dx \quad \text{if} \quad \int_0^\infty x |V(x)|\, dx<\infty \quad \text{and} \quad d + 2 \ell -2\geq 1.
\end{eqnarray}
In the case $d=2$ and $\ell = 0$, a stronger condition is imposed on the potential being the upper bound of the inequality  different \cite{seto1974bargmann}.
The previous inequality is optimal in the sense that there exits potentials $V_\varepsilon$ such that
\begin{equation*}
n^{d}_\ell+\varepsilon = \dfrac{1}{2 \,\ell +d-2}\int_0^\infty x |V_\varepsilon(x)|\, dx, 
\end{equation*}
being $\varepsilon$ an arbitrary small positive real number. In fact, when the potential is a linear combination of Dirac-$\delta$ potentials sufficiently distant from each other the last equality holds with $\varepsilon\to 0$ \cite{schwinger1961bound}.
The following result also matches with the properties of such potentials \cite{galindo1990quantum}.
\begin{prop}\label{res:lmax}
The $d$-dimensional quantum system described by the Hamiltonian \eqref{5} admits bound states with angular momentum $\ell$ if, and only if,
\begin{equation}
\ell_{max}\neq L_{max},\quad\text{and}\quad \ell \in \{0,1,\dots,\ell_{max} \}\quad (\ell_{max}>-1),
\end{equation}
where 
\begin{equation}\label{eq:lmax}
\ell_{max}\equiv \left\lfloor L_{max}\right\rfloor,\quad L_{max}\equiv \frac{ w_1- {x_0\, w_0}/2}{w_1^2+1}+\dfrac{2 - d}{2},
\end{equation}
being $\lfloor \cdot\rfloor$ the integer part.
%If $\ell_{max}\leq -1$ there are not bound states at all. 
In addition, if  $\lambda_\ell = -\kappa^2_\ell$ is the energy of the bound state with angular momentum $\ell$ the following inequality holds
\begin{equation*}
\lambda_\ell <\lambda_{\ell+1}< 0,\quad \ell \in \{0,1,\dots,\ell_{max} -1\}.
\end{equation*}
\end{prop}
\begin{proof}
We analyse the behaviour of $F(y_0)$ near the origin of energies.
The  solutions of (\ref{eq:RadialB}) satisfy 
\begin{eqnarray*}
\lim_{\kappa\to 0^+} \dfrac{d}{d\kappa}\log \mathcal{I}_\ell(\kappa x_0)=\frac{\ell}{x_0},\quad
\lim_{\kappa\to 0^+} \dfrac{d}{d\kappa}\log \mathcal{K}_\ell(\kappa\, x_0)=-\frac{d+\ell-2}{x_0},
\end{eqnarray*}
therefore the secular equation (\ref{eq30}) for $\kappa\to 0^+$  becomes
\begin{equation}\label{eq:Eqlmax}
F_{0}(\ell)\equiv\lim_{\kappa\to 0^+}F(y_0)=\alpha^{-1}({d+\ell-2})+\alpha  \ell=\left(\alpha -{\alpha^{-1} }\right) (d-2)+{\tilde \beta}  x_0.
\end{equation}
The function $F_{0}(\ell)$ is a strictly monotone function, increasing if $\alpha>0$ and decreasing if $\alpha<0$. In addition, from (\ref{eq:SignL1}) we can conclude that there are no bound states for $\ell>\ell_{max}$.
Using the definitions of  \eqref{eq:Notation}, the solution of (\ref{eq:Eqlmax}) is $L_{max}$ given by \eqref{eq:lmax}. Finally, with the previous analysis it is clear that $\kappa_\ell >\kappa_{\ell+1}$. 
\end{proof}
From \eqref{eq:lmax}, it can be seen that as the dimension of the system increases, the maximum angular momentum reached by the system decreases. This happens because the centrifugal potential in (\ref{16}) becomes more repulsive as $d$ grows.
In Fig.\ref{fig:BoundStates} we plot two configurations in two dimensions which illustrate the results of Propositions \ref{res:EnergyAngularMomenta} and \ref{res:lmax}. 
\begin{figure}[h!]
    \centering
    \begin{subfigure}[b]{0.50\textwidth}
        \includegraphics[width=1\textwidth]{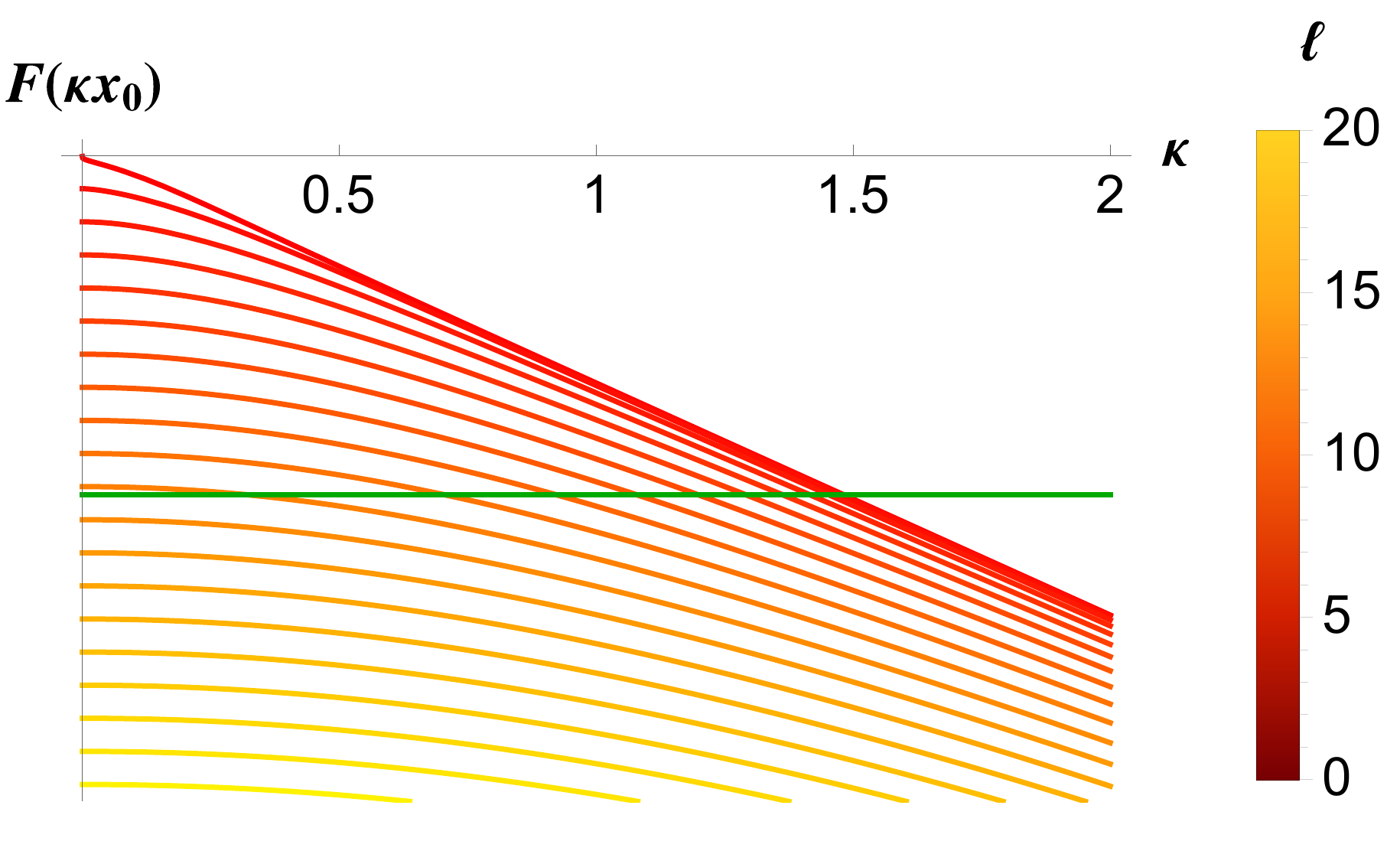}
    \label{fig:1BoundState1}
    \end{subfigure}
    \begin{subfigure}[b]{0.46\textwidth}
        \includegraphics[width=1\textwidth]{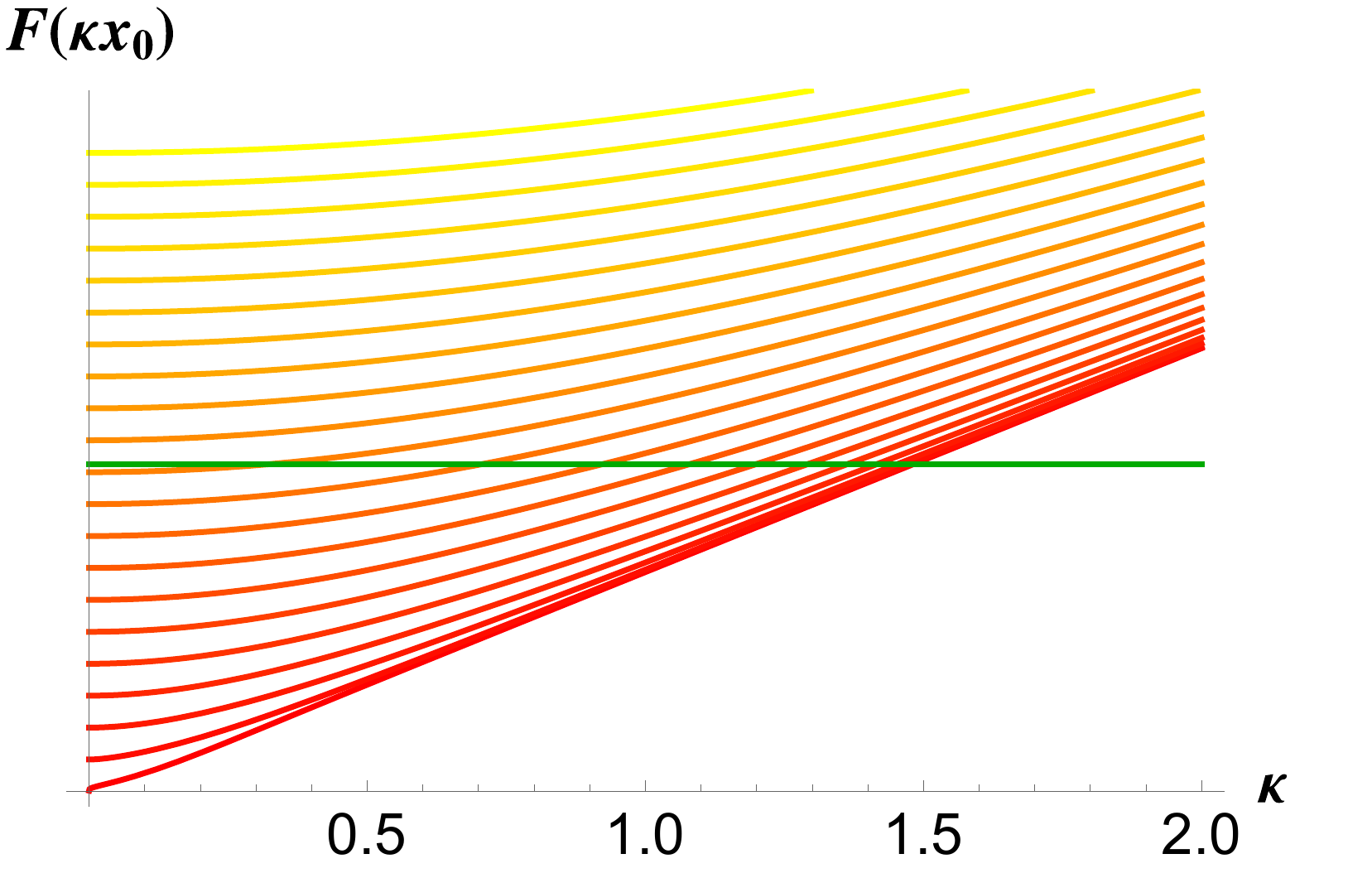}
    \label{fig:1BoundState2}
    \end{subfigure}
    \caption{Each curve represents $F(\kappa\,x_0)$ in \eqref{eq30} for different values of the angular momentum. The green horizontal line is the r.h.s. of \eqref{eq30}. LEFT: $d=2$, $\alpha =0.8$, ${\tilde \beta}=-3$ and $x_0 = 7$. RIGHT: $d=2$, $\alpha =-0.8$, ${\tilde \beta}=3$ and $x_0 = 7$.}
    \label{fig:BoundStates}
\end{figure}

The results obtained in \cite{govindarajan2016modelling} for $d=2$ and $d=3$ are recovered when $\vert w_0\vert x_0\gg \vert w_1\vert$ (for $d = 3$  there is minus sign and  the integer part  missing). To end this section, let us briefly study the behaviour of the number of negative energy eigenvalues as a function of the dimension $d$ and the angular momentum $\ell$. As was shown above, the number of bound states depends on deg$(d,\ell)$  \eqref{9}.  The increments with respect to $d$ and  $\ell$ are
\begin{eqnarray*}
\text{deg}(d+1,\ell)-\text{deg}(d,\ell)&=&\frac{ (d+\ell-3)!(d+2\ell-3)}{(\ell-1)! (d-1)!},\\[0.5ex]
\text{deg}(d,\ell+1)-\text{deg}(d,\ell)&=&\frac{ (d+\ell-3)!(d+2\ell-1)}{(\ell+1)! (d-3)!},
\end{eqnarray*}
therefore, both quantities are positive if $d\geq 3$ and $\ell\geq 1$. This ensures the growth of the number of bound states with the dimension and the angular momentum, except for $\ell=0$ where the degeneracy is always 1 (ground state) and for $d=2$ where $\text{deg}(2,\ell)=2$ for $\ell\geq1$. 
\subsection{Special feature of two dimensions}\label{sp-feat}
It is  known  that  the existence of bound states with $V_\delta=w_0\delta(x-x_0)$ necessarily imposes $w_0<0$ for any dimension $d$. This fact can be easily proved with the results obtained above. The maximum angular momentum for this potential is
\begin{equation}
\ell_{max}\equiv\left\lfloor \frac{- x_0\, w_0}{2}+\dfrac{2 - d}{2}\right\rfloor \leq L_{max} = \frac{ -{x_0\, w_0}}{2}+\dfrac{2 - d}{2} <  \dfrac{2 - d}{2}  \leq 0 \quad\text{if}\quad w_0>0,
\end{equation}
which means that there are no bound states if $w_0>0$. The next proposition shows that this condition on the coupling $w_0$ does not remain valid for all the cases when we add the $\delta^\prime$-potential, allowing the existence of a bound state with arbitrary positive $w_0$ for $d=2$ with $\ell=0$. This result is quite surprising taking into account  the usual interpretation of the Dirac-$\delta$ potential as a infinitely thin potential barrier if $w_0>0$. The key point to understand it is that only for $d=2$ and $\ell=0$ the centrifugal potential in \eqref{16} is attractive (centripetal), since $d+2\ell-3=-1<0$.
\begin{prop}\label{res:newBoundState} 
The quantum Hamiltonian \eqref{5} admits a bound state for any $w_0>0$  only if $d=2$ and $\ell=0$.
\end{prop}
\begin{proof}
From Proposition \ref{res:lmax} we conclude that
\begin{equation*}
L_{max} = \frac{1}{2} \left(2-d-\frac{x_0\, w_0}{w_1^2+1}+\frac{2 w_1}{w_1^2+1}\right)\leq 1/2 \qquad \text{if} \qquad w_0\geq 0,
\end{equation*}
since $2w_1/(1+w_1^2)\in[-1,1]$  $w_1\in\mathbb{R}$.
Therefore,  bound states with $\ell\geq 1$ are not physically admissible.
% In fact, the equality is reached in two dimensions with a single $\delta'$ interaction and $w_1 =1$.
For higher dimensions this state can not be achieved since
\begin{equation*}
L_{max}\leq 0 \qquad \text{if} \qquad d\geq3.
\end{equation*}
The equality is reached only if $w_0=0$, $w_1=1$ and $d=3$ being $\ell_{max} = L_{max} =0$. In this  case the selfadjoint extension of $\hvu$ which defines the potential $V_{\delta\text{-}\delta'}$ can not be characterised in terms of the matching conditions \eqref{eq:MatchingR}. In conclusion, with $V_{\delta\text{-}\delta'}$ described by (\ref{eq:MatchingR}) this bound state appears only if $d=2$ and $\ell=0$.

%\footnote{It leads to Robin boundary contidions at $x=x_0^-$ and Dirichlet boundary conditions at $x=x_0^+$ as it is shown in \eqref{eq:RobinDirichlet} (see \cite{munoz2015delta} for a detailed explanation).}.
%If $w_0=0$, $w_1=1$, $d=3$ then $L_{max}=\ell_{max}=0$ which results in $\kappa=0$.
%If $d=3$, $w_0 = 0$ and $w_1 = 1$,  $L_{max}(V_{\delta'})=0$, but this is a zero energy state, $\lambda_{\ell=0}=0$. In the previous situation $L_{max} = \ell_{max} = 0 $ and as it is proved in Result \ref{res:EnergyAngularMomenta} and shown in Fig. \ref{fig:BoundStates} this only occurs if  the r.h.s of (\ref{eq:SecularN=1}) intersects $F(y_0)$ at $\kappa_{\ell=0} =0$.
\end{proof}
%This bound state can appear for arbitrary positive values of $w_0$, not only if $w_0 \gtrsim 0$ as it could be expected, provided that $2\,w_1 >{x_0\, w_0}$. 
It is of note that the condition $2\,w_1 >{x_0\, w_0}$ ensures the existence of this bound state for arbitrary $w_0>0$. In addition, we must mention that the appearance of such bound state is significant because of two reasons. In one dimension, and with the definition of the $\delta'$ given by (\ref{eq:MatchingU}), this potential can not introduce bound states by itself \cite{gadella2009bound}. Furthermore, when $w_0>0$ the Dirac-$\delta$ potential $w_0\delta(x-x_0)$ can be interpreted as  an infinitely thin barrier, which contributes to the disappearance of bound states from the system. The result from Proposition \ref{res:newBoundState} is illustrated in Fig.\ref{fig:lmax2D3D}. At the end of the next section we will compute some numerical results that point out more differences with respect to the one dimensional $V_{\delta\text{-}\delta'}$ potential.
\begin{figure}[h!]
    \centering
    \begin{subfigure}[b]{0.544\textwidth}
        \includegraphics[width=1\textwidth]{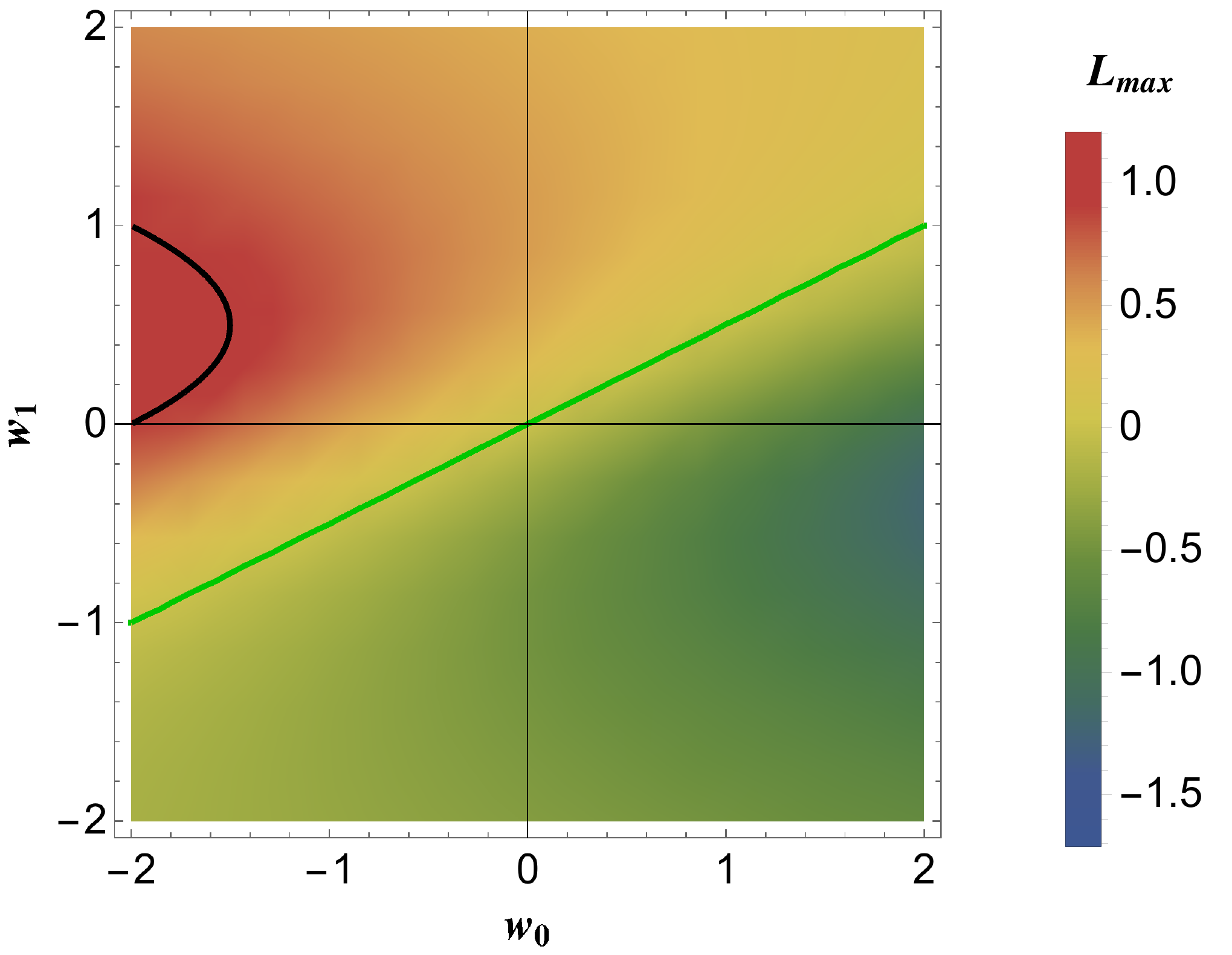}
    \label{fig:lmax2D}
    \end{subfigure}
    \begin{subfigure}[b]{0.445\textwidth}
        \includegraphics[width=1\textwidth]{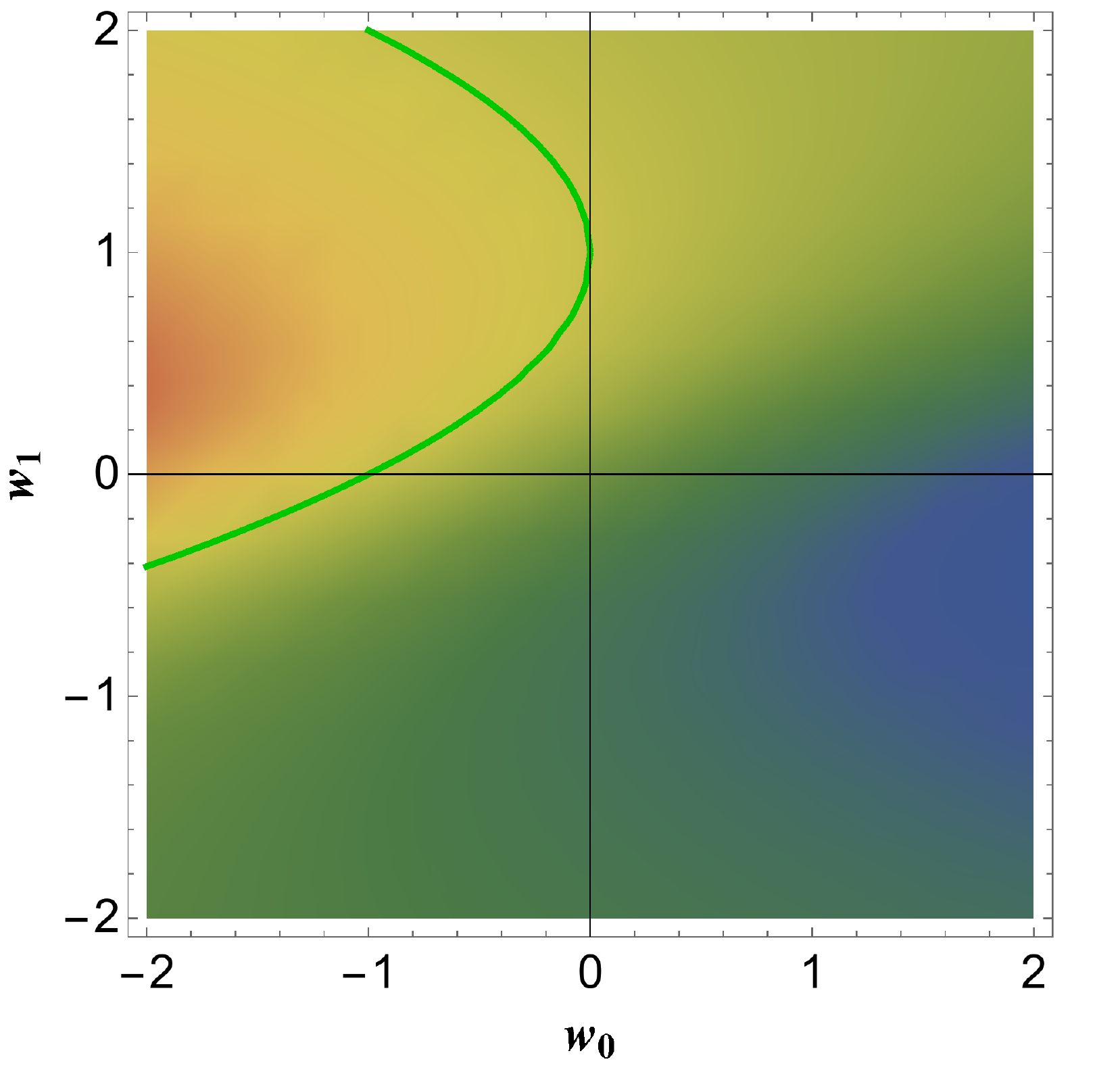}
    \label{fig:lmax3D}
    \end{subfigure}
    \caption{Plots of  $L_{max}$  \eqref{eq:lmax}, with $x_0=1$, as a function of $w_0$  and $w_1$. LEFT: $d=2$ showing  {{ $L_{max} = 0$}} (green line) and  {{$L_{max} = 1$}} (black curve). RIGHT: $d=3$ showing {{$L_{max} = 0$}} (green curve). There is a bound state with $\ell=0$ in two dimensions  for $w_0>0$, but this is not the case  in  three or higher dimensions.}
    \label{fig:lmax2D3D}
\end{figure}

%%%%%%%%%%%%%%%%%%%%%%%%%%%%%%%%%%%%%%%%%%%
\section{Scattering states, zero-modes, and some numerical reuslts}\label{zmode-sec}
%%%%%%%%%%%%%%%%%%%%%%%%%%%%%%%%%%%%%%%%%%%

%%%%%%%%%%%%%%%%%%%%%%%%%%%%%%%%%%%%%%%%%%%
\subsection{Scattering States}\label{sec:Scattering}
%%%%%%%%%%%%%%%%%%%%%%%%%%%%%%%%%%%%%%%%%%%

To complete the general spectral study of the potential \eqref{eq:PotentialDeltaDeltap} it is necessary to characterise its positive energy states, i.e., the scattering states. These states are always present in the system weather there exist negative energy states or not. In addition, when the parameters $w_0$ and $w_1$ are such that the potential $V_{\delta\text{-}\delta^\prime}$ does not admit bound states, the Schr\"odinger Hamiltonian \eqref{5} can be re-interpreted as the one particle states operator of an effective quantum field theory (see e.g. \cite{munoz2015delta,JMG} and references therein), where the scattering states are the one particle states of the scalar quantum vacuum fluctuations produced by the  field. With this interpretation, the explicit knowledge of the scattering states, specially the phase shifts, enables to obtain one loop calculations of the quantum vacuum energy acting on the internal and external wall of the singularity at $x=x_0$ \cite{kirstenbordag2014jpa}.
In this case, defining\footnote{By using this definition we recover the usual relation between $k$ (scattering states) and $\kappa$ (bound sates): $k\to i\kappa$ as we go from $\lambda>0$ to $\lambda<0$.} $k=\sqrt{\lambda}>0$, the general solution of  \eqref{eq11} is 
\begin{eqnarray}\label{eq:RadialScattering}
{R}_{k\ell}(x)   =\begin{cases}A_1\,\mathcal{J}_{\ell}(kx) +B_1\,\mathcal{Y}_{\ell}(kx) \quad \text{if}\quad x\in (0, x_0),\\[0.5ex] 
A_2\,\mathcal{J}_{\ell}(kx) +B_2\,\mathcal{Y}_{\ell}(kx) \quad \text{if}\quad x\in (x_0, \infty),
\end{cases}
\end{eqnarray}
being $\mathcal{J}_{\ell}(z)$ and $\mathcal{Y}_{\ell}(z)$, up to a constant factor, the  hyperspherical Bessel functions of the first and second kind respectively
\begin{equation*}
\mathcal{J}_{\ell}(z)\equiv \dfrac{1}{z^{\nu}} J_{\ell+\nu}(z), \quad \mathcal{Y}_{\ell}(z)\equiv \dfrac{1}{z^{\nu}} Y_{\ell+\nu}(z).
\end{equation*}

\begin{prop}
The phase shift $\delta_\ell(k)$ for the $\ell$-wave in a $d$-dimensional system described by a central potential with finite support $V$ is given by
\begin{equation}\label{eq:TangentPhaseShiftsGeneral}
\tan{\delta_\ell(k,V)} = -{B_{ext}}/{A_{ext}},
\end{equation}
where $A_{ext}$ and $B_{ext}$ are defined from the asymptotic behaviour of the radial function as
\begin{equation}
{R}_{k\ell}(x)\underset{x\to\infty}{\sim} x^{\frac{1-d}{2}}\left(A_{ext} \cos \mu_\ell + B_{ext} \sin \mu_\ell \right),\quad \mu_\ell \equiv k x-\frac{\pi }{2}  (\ell+\nu+\frac{1}{2}).
\end{equation}
\end{prop}
\begin{proof}
Far away from the origin the central potential  is identically zero, consequently the scattering solution will be a linear combination of $\mathcal{J}_{\ell}(kx)$ and $\mathcal{Y}_{\ell}(kx)$ which satisfy
\begin{eqnarray*}
\mathcal{J}_{\ell}(kx) \underset{x\to\infty}{\sim}  \sqrt{\frac{2}{\pi }} (kx)^{\frac{1}{2}-\frac{d}{2}} \cos \mu_\ell ,\quad
\mathcal{Y}_{\ell}(kx)\underset{x\to\infty}{\sim} \sqrt{\frac{2}{\pi }} (kx)^{\frac{1}{2}-\frac{d}{2}} \sin \mu_\ell .
\end{eqnarray*}
On the other hand, from partial wave analysis, the asymptotic behaviour of ${R}_{k\ell}(x)$ is proportional to $\cos \left(\mu_\ell + \delta_\ell \right)$ \cite{taylor-scatt},
where $\delta_\ell$ is the phase shift for the $\ell$-wave. Gathering both equations, 
\begin{eqnarray*}
A_{ext} \cos \mu_\ell + B_{ext} \sin \mu_\ell = C_{ext}\, \cos \left(\mu_\ell + \delta_\ell \right),
\end{eqnarray*}
from which the result \eqref{eq:TangentPhaseShiftsGeneral} is obtained.
\end{proof}
The previous result can be easily generalised to central potentials satisfying $x^2\,{V(x)}{}\to 0$ as $x\to\infty$ (see \cite{taylor-scatt}).
For the potential $V_{\delta\text{-}\delta^\prime}$, the square integrability condition on the radial wave function in any finite region sets $B_1=0$, except for $d=2$ and $\ell=0$ where the argument developed in Section \ref{sec:BS}, imposes $B_1=0$ \cite{berry1973semiclassical}.
In this way, using \eqref{eq:MatchingR} and \eqref{eq:RadialScattering} the exterior coefficients  $\{A_2,\,B_2\}$ can be expressed as
\begin{eqnarray}\label{eq:ExteriorCoefficients}
&&\left(
\begin{array}{c}
 A_2 \\
 B_2 \\
\end{array}
\!\right)= A_1 \left(
 \begin{array}{cc}
 \mathcal{J}_{\ell}(kx_0) & \mathcal{Y}_{\ell}(kx_0) \\ [0.5ex]
 k \,\mathcal{J}'_{\ell}(kx_0) & k \,\mathcal{Y}'_{\ell}(kx_0) \nonumber  \\[0.5ex]
\end{array}
\right)^{-1}\left(
 \begin{array}{cc}
 \alpha & 0 \\ [0.5ex]
 {\tilde \beta}  & \alpha^{-1} \nonumber  \\[0.5ex]
\end{array}
\right) \left(\!
\begin{array}{c}
 \mathcal{J}_{\ell}(kx_0) \\ 
k\mathcal{J}'_{\ell}(kx_0) \\
\end{array}
\!\right)\!\\
&&=\frac{1}{2 k} \pi  (k x_0)^{d-1}A_1 \left(
 \begin{array}{cc}
 k\mathcal{Y}'_{\ell}(kx_0) & -\mathcal{Y}_{\ell}(kx_0) \\ [0.5ex]
 -k \,\mathcal{J}'_{\ell}(kx_0) & \mathcal{J}_{\ell}(kx_0)  \nonumber  \\[0.5ex]
\end{array}
\right)\left(
 \begin{array}{cc}
 \alpha & 0 \\ [0.5ex]
 {\tilde \beta}  & \alpha^{-1} \nonumber  \\[0.5ex]
\end{array}
\right) \left(\!
\begin{array}{c}
 \mathcal{J}_{\ell}(kx_0) \\ 
  k\mathcal{J}'_{\ell}(kx_0) \\
\end{array}
\!\right)\!.
\end{eqnarray} 
From this result and \eqref{eq:TangentPhaseShiftsGeneral} we get
\begin{equation}\label{38}
\tan{\delta_\ell(k,V_{\delta\text{-}\delta^\prime})}=-\frac{\mathcal{J}_{\ell}(kx_0) \left(\left(1-\alpha ^2\right) k \mathcal{J}'_{\ell}(kx_0)+\alpha  {\tilde \beta} \mathcal{J}_{\ell}(kx_0)\right)}{- k  \mathcal{J}^\prime_{\ell}(kx_0) \mathcal{Y}_{\ell}(kx_0)+\mathcal{J}_{\ell}(kx_0) \left(\alpha^2 k\mathcal{Y}'_{\ell}(kx_0)-\alpha  {\tilde \beta} \mathcal{Y}_{\ell}(kx_0)\right)}.
\end{equation}
In the spherical wave basis, the scattering matrix is diagonal and its eigenvalues can be written as
\begin{equation}\label{smat}
\exp\left({2i\delta_\ell(k,V_{\delta\text{-}\delta^\prime})}\right)=({1+2 i \tan \delta_\ell-\tan ^2\delta_\ell  })/({1 +\tan ^2\delta_\ell}).
\end{equation}
Note that for the potential $V_{\delta\text{-}\delta^\prime}$, the secular equation \eqref{eq30} can be re-obtained as the positive imaginary poles of \eqref{smat} using \eqref{38} (for details see \cite{taylor-scatt}). 

To complete this section, let us show explicit formulas of the phase shift for some particular cases of the potential $V_{\delta\text{-}\delta^\prime}$ previously studied in the literature:
\begin{itemize}
\item The $\delta$-potential ($w_1=0,\tilde{\beta}=w_0$) phase shift is
\begin{equation*}
\tan{\delta_\ell(k,V_{\delta})}=\frac{\pi \, w_0\, x_0\, J_{\ell +\nu }\left(k\,x_0\right){}^2}{\pi \,  w_0\, x_0\, J_{\ell +\nu }\left(k\,x_0\right) Y_{\ell +\nu }\left(k\,x_0\right)-2},
\end{equation*}
which matches for $d=2,3$ with the results obtained in \cite{lapidus1986scattering} and \cite{kirstenbordag2014jpa} respectively. 
\item The hard hypersphere defined as
\begin{equation*}
V_{hh}(x)=\begin{cases}
\infty, & {x\leq x_0},\\
0, & {x>x_0},
\end{cases}
\end{equation*}
imposes Dirichlet boundary conditions for the wave function  on the exterior region, $R(x_0^+) = 0$. The same result can be obtained from the $\delta\text{-}\delta^\prime$ potential setting  $w_1 \to -1$  (\ref{eq:RobinDirichlet}). Thus, the phase shift is
\begin{equation*}
\tan{\delta_\ell(k,V_{hh})}=\lim_{w_1\to-1}\tan{\delta_\ell(k,V_{\delta\text{-}\delta^\prime})}=\frac{J_{\ell +\nu }\left(k x_0\right)}{Y_{\ell +\nu }\left( k x_0 \right)}.
\end{equation*}
For two and three dimensional systems it coincides with \cite{lapidus1986scattering} (hard circle) and \cite{huang1957quantum} (hard sphere) respectively. 
\item When we turn off the Dirac-$\delta$ term ($w_0=0\Rightarrow\beta=0$) we have that there is an effective $\delta$ potential coupling characterised by
\begin{equation*}
{\tilde \beta}=-\frac{(\alpha-\alpha^{-1})(d-1)}{2x_0},
\end{equation*}
therefore from \eqref{38} we obtain that the phase shift for the pure $\delta^\prime$ is
\begin{equation*}
\tan{\delta_\ell(k,V_{\delta^\prime})}=
-\frac{\left(1-\alpha ^2\right) \mathcal{J}_{\ell}(z_0) ((d-1) \mathcal{J}_{\ell}(z_0)+2z_0 \mathcal{J}'_{\ell}(z_0) )}{\left(\alpha
   ^2-1\right) (d-1) \mathcal{J}_{\ell}(z_0) \mathcal{Y}_{\ell}(z_0)+2 z_0\left( \alpha ^2  \mathcal{Y}'_{\ell}(z_0)\mathcal{J}_{\ell}(z_0) -\mathcal{J}'_{\ell}(z_0)  \mathcal{Y}_{\ell}(z_0)\right)},
\end{equation*}
where $z_0\equiv kx_0$. As can be seen, $\delta_\ell(k,V_{\delta^\prime})$ depends on the energy through $z_0$ unlike it happens with the scattering amplitudes for the pure $\delta^\prime$ potential in one dimension, where there is no dependence on the energy \cite{gadella2009bound,munoz2015delta,gadella-jpa16}. Nevertheless, what is maintained is the conformal invariance of the system, i.e., the phase shift is invariant under
\begin{equation}
x_0\to\Lambda x_0,\quad k\to\frac{k}{\Lambda}, \quad w_1\to w_1.
\end{equation}
\end{itemize}

\subsection{On the existence of zero-modes}
In this section we will deduce the conditions which ensure the existence of states with zero energy for the $\delta\text{-}\delta^\prime$ potential. The presence of an energy gap between the discrete spectrum of negative energy levels and the continuum spectrum of positive energy levels is of great importance in some areas of fundamental physics (see, e.g. \cite{qft-math}), specially when we promote non-relativistic quantum Hamiltonians to effective quantum field theories under the influence of a given classical background.
To start with, we solve (\ref{eq:ReducedRadialEq}) for $\lambda= 0$
\begin{equation}\label{eq:ReducedRadialEqZeroMode}
\left[ \dfrac{d^2}{dx^2}\!  -\frac{(2-\eta) (4-\eta)}{4 x^2} \right]u_{0\,\ell}(x)=0\quad\text{with} \quad \eta \equiv 5 - ( d + 2\ell).
\end{equation}
The general solution of the zero-mode differential equation is given by
\begin{equation}\label{eq:SolutionZeroMode}
v_\eta(x)\equiv u_{0\,\ell}(x)=\begin{cases}c_1 \, x^{\frac{\eta-2}{2}} + c_2 \, x^{\frac{4-\eta}{2}} & \text{if}  \quad \eta\neq 3,
\\[0.5ex] 
c_1 \sqrt{x} + c_2 \sqrt{x}\, \log x \quad  & \text{if}  \quad\eta=3.
\end{cases}
\end{equation}
It must be emphasized that $\eta=3$ corresponds to  $d=2$ and $\ell=0$.
%\begin{equation}\label{eq:SolutionZeroMode}
%v_\eta(x) \equiv u_{0\,\ell}(x)=c_1 \, x^{\frac{\eta-2}{2}} + c_2 \, x^{\frac{4-\eta}{2}}.
%\end{equation}
%When $\eta = 3$, i.e. $d=2$ and $\ell=0$ the two functions of (\ref{eq:SolutionZeroMode}) are not linearly independent, indeed they are equal, and the general solution is
%\begin{equation*}
%v_{\eta=3}(x)= c_1 \sqrt{x} + c_2 \sqrt{x}\, \log x.
%\end{equation*}
In order to determine the integration constants of the general solution \eqref{eq:SolutionZeroMode}  we must impose two requirements. The first condition is 
square integrability 
\begin{equation}\label{eq:ZeroModeIntegrability}
\int_0^{\infty} |v_\eta(x)|^2 \, dx=\int_0^{x_0} |v_\eta(x)|^2 \, dx + \int_{x_0}^{\infty} |v_\eta(x)|^2 \, dx < \infty,
\end{equation}
where both integrals should be finite. The second one is the matching condition that defines the $\delta\text{-}\delta^\prime$ singular potential \eqref{eq:MatchingU}. Depending on $\eta$, i.e., the angular momentum $\ell$ and the dimension of the physical space $d$, we can distinguish two cases.

\noindent {\bf Case 1:} $\boldsymbol{\eta \in \{1,2, 3\}}.$
After imposing \eqref{eq:ZeroModeIntegrability} we end up with the reduced radial wave functions
\begin{eqnarray*}
v_\eta(x)=
\begin{cases}
 \sqrt{x}\,(c_1  + c_2\, \, \log x)& \text{if}\quad \eta=3,    \\[0.5ex]
 \hfil c_1+c_2\,x& \text{if}\quad \eta=2, \\[0.5ex]
 \hfil c_1 x^{3/2}& \text{if}\quad \eta=1,
\end{cases}
\quad \text{for} \quad x<x_0 \quad \text{and} \quad v_\eta(x) = 0 \ \  \text{for} \ \  x>x_0.
\end{eqnarray*}
In this case the matching conditions of (\ref{eq:MatchingU}) are satisfied if, and only if, $c_1=c_2=0$. Therefore there are no zero energy states.

\noindent {\bf Case 2:} $\boldsymbol{\eta \leq 0}.$ In this situation, the square integrable solution  and matching conditions  result in
\begin{equation}\label{52}
v_\eta(x) = \begin{cases}
c_2\,x^{\frac{4-\eta}{2}} & \ {x< x_0},\\[0.5ex]
c_1\,x^{\frac{\eta-2}{2}} &\ {x>x_0},
\end{cases};
\ \  
c_1 \left(
\begin{array}{c}
 x_0^{\frac{\eta-2}{2}} \\
 \frac{\eta-2}{2}  x_0^{\frac{\eta-2}{2}-1}  \\
\end{array}
\right)=c_2\left(
\begin{array}{cc}
 \alpha  & 0 \\
 \beta  & {\alpha^{-1} } \\
\end{array}
\right)\left(
\begin{array}{c}
 x_0^{\frac{4-\eta}{2}}  \\
 \frac{4-\eta}{2} x_0^{-\frac{\eta-2}{2}} \\
\end{array}
\right)\!.
\end{equation}
A non trivial solution exists if, and only if, the system satisfies 
\begin{equation}\label{eq:BetaNonTrivial}
\beta = \frac{-2 \alpha ^2+\alpha ^2 \eta+\eta-4}{2 \alpha  x_0}\quad \text{with}\quad c_2= {x_0^{\eta-3}}\alpha^{-1}c_1.
\end{equation}
In addition, the regularity condition at $x=0$ is also satisfied: $v_\eta(x=0)=0.$
Indeed, if the previous equation is inserted in \eqref{eq:lmax} we obtain
\begin{equation*}
L_{max}=\ell_{max}=\ell,
\end{equation*}
which is in agreement with our previous analysis of the energy levels, i.e., if $L_{max}=\ell_{max}$ the left hand side of the secular equation \eqref{eq30}, $F(\kappa_\ell \, x_0)$, reaches the right hand side at $\kappa_\ell = 0$. The reverse is also true.
\subsection{The mean value of the position operator}
In this section we will show some numerical results concerning the expectation value of $x$ for the bound states that satisfy $\eta< 0$  (as a function of the parameters $w_0$ and $w_1$). Once the dimension $d$, the radius $x_0$ and the angular momentum $\ell$ are fixed, the plane $w_0$-$w_1$ is divided into two zones: one in which the bounds states do not exist and another one in which they do. The limit between these two zones corresponds to the zero-mode states\footnote{This is ensured by the condition $\eta\leq 0$. If $\eta>0$ the limit between the two zones does not correspond to a physically meaningful state as it was previously demonstrated.}. The existence  of zero-modes is of critical importance to compute numerically the expectation value of the dimensionless radius $x$ when the parameters $w_0$ and $w_1$ are close to the common boundary of the regions mentioned.

For a given bound state of energy $\lambda_\ell=-\kappa_\ell^2$, the general expression for the expectation value $
\langle x\rangle_{\kappa\ell}\equiv \langle \Psi_{\kappa\ell} |\, x \,|\Psi_{\kappa\ell}\rangle  
$
is given in terms of the reduced radial wave function as
\begin{equation}\label{47}
\hspace{-0.1cm}\langle x\rangle_{\kappa\ell}= 
%\dfrac{\displaystyle\int_0^\infty  x |u_{\kappa\ell}(x)|^2dx}{\displaystyle\int_0^\infty  |u_{\kappa\ell}(x)|^2 dx }=
\dfrac{1}{\kappa_\ell}\dfrac{\displaystyle\int_0^{\kappa_\ell\,x_0} \, z^2 I^2_{\ell+\nu}( z)dz +\left(\dfrac{\alpha\,I_{\ell+\nu}( \kappa_\ell x_0)}{K_{\ell+\nu}(\kappa_\ell x_0)}\right)^2 \displaystyle\int_{\kappa_\ell\,x_0}^{\infty} \, z^2 K^2_{\ell+\nu}(z)dz}{\displaystyle\int_0^{\kappa_\ell\,x_0} \, z I^2_{\ell+\nu}(z)dz + \left(\dfrac{\alpha\,I_{\ell+\nu}(\kappa_\ell x_0)}{K_{\ell+\nu}(\kappa_\ell x_0)}\right)^2\displaystyle\int_{\kappa_\ell\,x_0}^{\infty} \, z K^2_{\ell+\nu}(z)dz}.
\end{equation}
The last expression does not depend explicitly on $\beta$  \eqref{eq:Notation}, but it does through $\kappa_\ell$. If we take the limit $\kappa_\ell\to 0^+$ we obtain
\begin{equation*}\label{eq:limitkappato0}
\dfrac{\langle x\rangle_{0\ell}}{x_0}\equiv\lim_{\kappa_\ell\to 0^+}\dfrac{\langle x\rangle_{\kappa\ell}}{x_0}=
\begin{cases}
\dfrac{\eta-1}{\eta}\left( 1-\left|\dfrac{2 (\eta-3)}{(\eta-6)  \left(\alpha ^2 (\eta-5)+\eta-1\right)}\right|\right) &\quad \eta\leq -1,\\[2.5ex] 
\hfil \infty& \quad  \eta\in\{0,1,2,3\}.
\end{cases}
\end{equation*}
As expected, this result coincides with the calculation of the mean value for the zero-modes, carried out with the wave functions in \eqref{52}.
As can be seen, when there exist zero-modes with $\eta<0$, the expectation value $\langle x\rangle_{0\ell}$ is finite, but when the system does not admit them, or $\eta=0$ the limit $\langle x\rangle_{0\ell}$ is divergent. Somehow, the zero-modes with $\eta=0$ are semi-bound states in the sense that the expectation value is divergent. This behaviour gives rise to three different situations:
\begin{itemize}
\item When there are zero-modes with $\eta<0$, the mean value $\langle x\rangle_{\kappa\ell}$ for the bound states has a finite upper bound 
\begin{equation}\label{56}
\lim_{w_1\to1}\frac{\langle x\rangle_{0\ell}}{x_0}=({\eta-1})/{\eta}.
\end{equation}
\item If there is a semi-bound zero-mode, i.e., $\eta=0$, the upper bound imposed by $\langle x\rangle_{0\ell}$ is infinite: $\langle x\rangle_{\kappa\ell}$ diverges as $\lambda_\ell\to 0^-$ in the $w_0$-$w_1$ plane.
\item When there are no zero-modes, $\langle x\rangle_{\kappa\ell}$ does not have an upper bound  and therefore, as $\lambda_\ell$ goes to zero  the expectation value goes to infinity. This fact can be interpreted as the state disappearing from the system: when $\lambda_\ell\to 0^-$ the corresponding wave function becomes identically zero.
\end{itemize}
In Fig.\ref{fig:MeanValuex} we have plotted the mean value of two configurations as a function of the couplings $w_0$ and $w_1$ for values of $d$ and $\ell$ such that $\eta<0$ (there is a zero-mode). We have distinguished the region in which the expectation value of $x$ lies outside the $\delta\text{-}\delta'$ horizon and the one with $\langle x\rangle<x_0$. The former, bearing in mind the original ideas by G. 't Hooft \cite{thooft1,thooft2},  would correspond to the states of quantum particles falling into a black hole that would be observed by a distant observer. Indeed, the amount of bound states for two and three dimensions is proportional to $x_0$ and $x_0^2$ respectively and as it is mentioned\footnote{Although the formulas for $\ell_{max}$ presented in \cite{govindarajan2016modelling} are only valid when \eqref{eq:X} is satisfied, the behavior of the total amount of bound states as a function of $x_0$ does not change (as long as $x_0$ is large enough). Consequently, the argument for the area law remains valid.} in \cite{govindarajan2016modelling}  these bound states would give an area law for the corresponding entropy in quantum field theory when they are interpreted as micro-states of the black hole horizon.
\begin{figure}[h!]
    \centering   
        \includegraphics[width=1\textwidth]{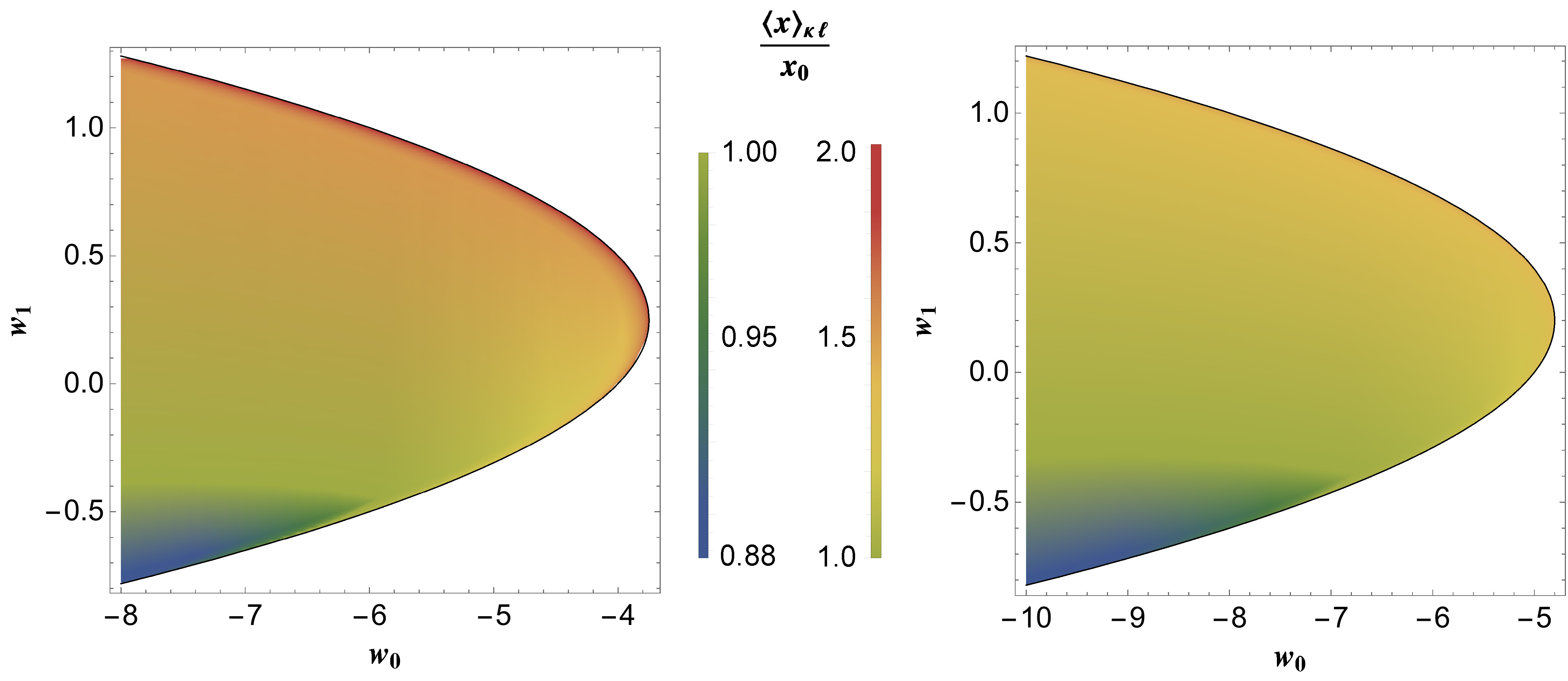}
    \caption{Mean value of the dimensionless radius operator $\langle x\rangle/x_0$ given in \eqref{47}. LEFT: $x_0 = 1,\  \ell= 2 \ \text{and}\ d=2$ being $\eta = -1$. RIGHT: $x_0 = 1,\  \ell= 2 \ \text{and}\ d=3$ being $\eta = -2$. The limit $\kappa\to 0^+$ in \eqref{47} fits with \eqref{eq:limitkappato0}. The black curve satisfies $L_{max} = \ell$ in each case.}
    \label{fig:MeanValuex}
\end{figure}
%\begin{figure}[h!]
%    \centering   
%        \includegraphics[width=1\textwidth]{ValorMedioRadioJuntas}
%    \caption{Mean value of the dimensionless radius operator $\langle x\rangle/x_0$ given in \eqref{47}. LEFT: $x_0 = 1,\  \ell= 2 \ \text{and}\ d=2$ being $\eta = -1$. RIGHT: $x_0 = 1,\  \ell= 2 \ \text{and}\ d=3$ being $\eta = -2$. The limit $\kappa\to 0^+$ in \eqref{47} fits with \eqref{eq:limitkappato0}. The black curve satisfies $L_{max} = \ell$ in each case.}
%    \label{fig:MeanValuex}
%\end{figure}
\subsubsection*{On the energy shifts produced by the $\delta'$}  It is worth mentioning two  central differences between the present analysis  in arbitrary dimension with the hyperspherical
 $\delta$-$\delta'$ potential ($d\geq 2$) and the one dimensional point analog \cite{gadella2009bound}.
 In the latter, the $\delta'$ by itself ($w_0=0$) only gives rise to a pure continuum spectrum of positive energy levels (scattering states). In addition, for the one dimensional case when $w_0<0$ the appearance of the $\delta^\prime$-term in the potential increases the energies of the bound states because it breaks parity symmetry, which does not happen for $d\geq 2$. These two properties are not maintained in general for $d\geq 2$. For example, in two
 dimensions there is a bound state with energy $\lambda_{\ell=0}=-1.205$ if $w_1=0.9$ ($x_0=0.15$ and, of course,
 $w_0=0$). Secondly, the previous case and all the study of Section \ref{sp-feat} prove that there are bound states with lower energy when
 the $\delta'$ is added to the $\delta$ potential. In view of the above, it could be thought that it only takes place when
 $w_0\geq 0$ (since the $\delta$ potential presents no bound states). However, if we consider a three dimensional system
 with $x_0=1$ and $w_0=-1.85$, a single bound state with energy $\lambda_{\ell=0}=-0.514$ appears when $w_1=0.437$ and
 with $\lambda_{\ell=0}=-0.482$ if we turn off the $\delta'$. What we can conclude from the numerical results is that the $\delta$-$\delta^\prime$ potential can give rise to a lower energy fundamental state than if it has only the $\delta$ potential for $d\geq 2$.

%%%%%%%%%%%%%%%%%%%%%%%%%%%%%%%%%%%%%%%%%%%
\section{Concluding remarks}\label{sec:conclu}
%%%%%%%%%%%%%%%%%%%%%%%%%%%%%%%%%%%%%%%%%%%

Our study provides novel results with $\delta\text{-}\delta^\prime$ hyperspherical potentials. Firstly, on the basis of this paper in arbitrary dimension, a careful study of the applications that we have  already reported
(and others) can be performed. The special attention paid on bound states is justified: as was shown in \cite{govindarajan2016modelling} the bound states can be thought of, in a quantum field theoretical view,  as photon states falling into a black hole for an observer far away from the event horizon. In this sense, the $\delta\text{-}\delta^\prime$ potential generalises the brick wall model by G. 't Hooft \cite{thooft1,thooft2}. In addition, the knowledge of the bound state spectrum of the system plays an essential role in the study of fluctuations around classical solutions and in the Casimir effect when the Sch\"odinger operator $-\Delta_d+V_{\delta\text{-}\delta^\prime}$ is reinterpreted as the one particle Hamiltonian of an effective quantum field theory. 

Our first achievement is the generalisation of the results given in \cite{kurasov1996distribution} for the one dimensional $\delta^\prime$-potential. We have introduced a rigorous and consistent definition of the potential $V_{\delta\text{-}\delta^\prime}=w_0\delta(x-x_0)+2w_1\delta^\prime(x-x_0)$ in arbitrary dimension, characterizing a selfadjoint extension of the  Hamiltonian $\hvu$ \eqref{16} defined on  $\mathbb{R}_{x_0}$. In doing so, we have corrected the matching conditons in \cite{govindarajan2016modelling} for the two and three dimensional $V_{\delta\text{-}\delta^\prime}$ potential. We have shown that the Dirac-$\delta$ coupling requires a re-definition which also depends on the radius $x_0$ and the $\delta'$ coupling $w_1$.

We have also characterised the spectrum of bound states in arbitrary dimension, computing analytically the amount of bound states for any values of the free parameters $w_0$, $w_1$ and $x_0$ that appear in the Hamiltonian. One of the most interesting and counterintuitive results we have found is the existence of a negative energy level for $d=2$ and $\ell=0$ when the Dirac-$\delta$ coupling $w_0$ is positive. In such a situation, the Dirac-$\delta$ potential $w_0\delta(x-x_0)$, with $w_0>0$, is an infinitely thin potential barrier, therefore bound states in the regime $w_0>0$ are not expected (as it happens for the one dimensional analog \cite{munoz2015delta}). 

As a limiting case of the spectrum of bound states for the Hamiltonian \eqref{5},
we have obtained the spectrum of zero-modes of the system in terms of the parameter $\eta=5-(d+2\ell)$. We have shown that the conditions on  $w_0$, $w_1$ and $x_0$ for the existence of zero-modes are
$
\ell_{max}=L_{max}\quad\text{and}\quad\eta\leq 0.
$
In addition, we have  computed numerically the expectation value $\langle x\rangle_{\kappa\ell}/x_0$ for the bound states with energy $\lambda=-\kappa^2$ and angular momentum  $\ell$ as a function of  $w_0$ and $w_1$. This calculation has enabled us to realise that the zero-modes with $\eta<0$ behave as bound states in the sense that
$
\langle x\rangle_{0\ell}<\infty
$, 
and the zero-modes corresponding to $\eta=0$ behave as semi-bound states due to
$
\langle x\rangle_{0\ell}=\infty
$. 
These results  determine the topological properties of the space of states of the system since the existence  of zero-modes characterises the space of couplings.
%\begin{itemize}
%\item If there are zero-modes the space of couplings is $\{(w_0,w_1)\in\mathbb{R}^2\}$
%\item If there are no zero-modes then $\{(w_0,w_1)\in\mathbb{R}^2\setminus\{\ell_{max}(w_0,w_1)=L_{max}(w_0,w_1)\}\}$
%\end{itemize}

To complete our study of the Hamiltonian \eqref{5} we have obtained an analytical expression for all the phase shifts which describe all the scattering states of the system. This calculation is of central importance when we promote \eqref{5} to an effective quantum field theory (see \cite{munoz2015delta}) under the influence of a classical background. In this scenario, the knowledge of the phase shifts allows us to compute the zero point energy \cite{kirstenbordag2014jpa}. In addition, as it is shown in \cite{kirstenbordag2014jpa} the phase shifts contain in their asymptotic behaviour all the heat kernel coefficients of the asymptotic expansion of the heat trace.

Further work for the future could usefully be to add a non-singular hyperspherical background potential $V_0(x)$ to the $V_{\delta\text{-}\delta'}(x)$. For example, the spectrum of  $|x|$ plus the  $\delta$-$\delta'$ potential at the origin is studied for one dimensional systems in \cite{fassari18}. For these cases, first of all, we would have to define the selfadjoint extension which characterizes $V_{\delta\text{-}\delta'}$ considering $\hvup\equiv \hvu + V_0(x)$ instead of $\hvu$. If $V_0(x)$ satisfies  the hypothesis of the Kato-Rellich theorem, the selfdadjointness of  $\hvup$ is guaranteed by the selfdadjointness of  $\hvu$ \cite{oliveira09}. In this way, for this kind of potentials it seems  reasonable that  the analysis carried out in Section \ref{sec:BS} can be generalised  by just exchanging the modified hyperspherical Bessel functions ($V_0=0$)  by the corresponding general solutions of the background potential $V_0(x)$. Of course, most  potentials can not be solved analytically \citep{albeverio2005solvable}, but it is worth exploring the (solvable)  Coulomb potential $ V_0(x) = -{\gamma}/{x}$ with $\gamma\in\mathbb{R}_{>0}$ in arbitrary dimension. In addition to its  known applications in a multitude of  disciplines, this potential has recently been shown to play a central role in condensed matter physics to mimic impurities in real graphene sheets and other two dimensional systems \cite{prl-coulomb,coulomb-imp,coulomb-imp2}. For the Coulomb potential, the general solution can be written in terms of Whittaker functions which are closely related to the modified Bessel functions studied in the free case \cite{richard2018radial,olver2010nist}. Some important differences with respect to the latter are expected, e.g.,  an infinite number of negative energy levels ($\ell_{max}\to\infty$) with, possibly, an accumulation point not necessarily at zero energy.

\section*{Acknowledgements}
Partial financial support is acknowledged to the Spanish Junta de Castilla y León and FEDER (Projects VA057U16 and VA137G18) and MINECO (Project MTM2014-57129-C2-1-P). We are very grateful to K. Kirsten, L. Santamar\'\i a-Sanz, M. Gadella, and J. Mateos Guilarte for fruitful discussions. JMMC acknowledges the hospitality and support received from the ETSIAE-UPM at Madrid.

%\section*{References}
%\bibliography{Bibliography}
%\bibliographystyle{elsarticle-num} %sort by appearance

\end{document}